\documentclass{amsart}
\usepackage{amsmath}%
\usepackage{amsfonts}%
\usepackage{amssymb}%
\usepackage{graphicx}
\newtheorem{theorem}{Theorem}
\theoremstyle{plain}

\newtheorem{corollary}{Corollary}

\newtheorem{lemma}{Lemma}

\numberwithin{equation}{section}
\begin{document}

\title[On the solution of Four Color problem]{The chromatic class and the chromatic number of the planar conjugated triangulation}
\author{Natalia L. Malinina}
\address[Moscow Aviation Institute (National Research University), 4, Volocolamskiy Shosse, Moscow, GCP-4, Russian Federation]
\newline%
\email[Natalia Malinina]{malinina806@gmail.com}%


\thanks{This paper is in preliminary form and it will be submitted for publication in some journal.}
\date{December, 2012}
\dedicatory{Dedicated to my father Leonid Malinin}

\begin{abstract}
This material is dedicated to the estimation of the chromatic number and class of the conjugated triangulation and of the second conversion of the planar triangulation. This estimation may through some light on the difficulties, connected with the proof of Four Color Problem. Also this paper introduces some new hypotheses, which are equivalent to Four Color Problem.
\end{abstract}
\maketitle

\section {Introduction}

The chromatic number gives us the minimal number of the colors with the help of which we will be able to color rightly the vertexes of the graph. It means that no adjacent vertexes will be colored in the same color. It is known that for the planar graph the chromatic number is not more than 5 \cite{Berge}. The task of the graph's chromatic number estimation is called as the task on the graph's coloring. The graph vertexes' coloring, which corresponds to chromatic number, divides the set of the graph's vertexes in the subsets, which number is equal to the chromatic number. Each of these subsets contains the vertexes of one color. These subsets appear to be independent, because in the range of one subset there are no adjacent vertexes. The chromatic number cannot be found on the base of knowing the numbers of both the vertexes and the edges of the graph. The task of the chromatic number's estimation for the arbitrary graph is the task of many researchers \cite{Chrisofides}.

\section {The estimation of the chromatic class and number of the planar conjugated triangulation}

Let us continue to examine the properties of the planar conjugated triangulation and move to the estimation of its chromatic number and class. Such estimation can be done with the help of Brooks' and Shannon's theorems. But let us begin with formulations of some other theorems. 

\begin {theorem} [Berge's Theorem]
Any planar graph appears to be 5-chromatic \cite{Berge}.
\end {theorem}

The lowest valuations of the chromatic number are more interesting than the upper valuations, because they can be used in the procedures for the estimation of the chromatic number, which include the search tree if they are off base to the true value. The upper valuations cannot be used in such cases. However, the upper valuation, which can be easily determined, may be developed with the help of Brooks theorem for the planar graphs. Other valuations can be found at Welsh \cite{Welsh} and \cite {Wilf}. But upper valuations do not have a lot of practical significance.

\begin {theorem} [Tutte's theorem]
Let $G$ --- be the graph without the loops. Then $P(G,n)\geq0$ for all the positive integer $n$ and at $n\geq|V(G)|$ this inequality appears to be the strict one. If $P(G,n)=0$ then $P(G,m)=0$ for the positive integer $m>n$. This number is named as the chromatic number of the graph.
\end {theorem}

Let us formulate Brook's theorem for the conjugated triangulation.

\begin {theorem} [Brooks' theorem]

Let graph $H(V,Q)$ meet the requirements \cite{Brooks}:
\begin {enumerate}
\item
For each $v_h\in V\rho_h(v)\leq\rho_h$, where $\rho_h\geq3$.
\item
No connectivity component of the graph $H$ is the complete $(\rho_{h+1})$-vertex graph $H_{\rho+1}$
\end {enumerate}
Here: $\rho_h$ --- is the maximal degree of the vertex. 
Then $\gamma(H)\leq\rho_h$.
\end {theorem}

Let us confine ourselves only to formulating Brook's theorem, because it is well known.
In our case, for the planar conjugated triangulation, when the degrees of the vertexes are not more than four, Brooks' theorem will be formulated in the following way:

The conjugated triangulations meet the following requirements:

\begin {enumerate}
\item
For each $v_h \in V[\rho_h(v)\leq \rho_h]$, where $\rho_h \geq 3$.
\item	
No connectivity component of graph $H$ is the complete 5-vertexes graph $F_5$.
\end {enumerate}
Then $\gamma(H)\leq4$.

Brooks' theorem gives us the upper estimation of the chromatic number for the planar conjugated triangulation. But as far as we are searching the proof of the Four Color Problem it is necessary to prove that the chromatic number of the planar conjugated triangulation must be: $\gamma(H)\leq3$.

Let us further introduce some new hypotheses, which is also equivalent to the Four Color Problem, notably:

\subsection {Hypothesis 10}\footnote{The number of the Hypothesis comes from the list of them \cite{Zukov}}.

The vertexes of the planar conjugated triangulation can be painted in three colors thus no adjacent vertexes will be colored equally.

In compliance with Brooks' theorem $\gamma(H)\leq4$. Brooks' theorem is correct at $\rho_h\geq3$. It is evident that it will be all the more correct at $\rho_h\geq2$.

Let us compose graph $G(Q,\Gamma)$, as the adjacent edge graph of graph $H$, complying the following rules (Fig.\ref {Fig:1}):

\begin {enumerate}
\item
Let us sign with a dot the centers of each graph's $H$ edge and accept these dots as the vertexes $q$ of graph $G$.
\item
Let us connect the graph's $G$ vertexes by the edges, if the corresponding edges have the common (shared) ending.
\end {enumerate}

\begin{figure}[htb]
	 \includegraphics[width=0.95\textwidth]{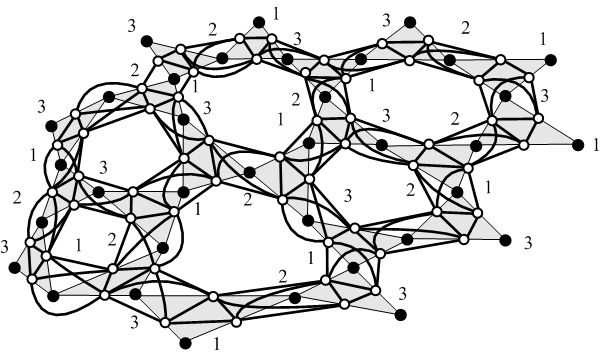}
	\caption{}
	\label{Fig:1}
\end{figure}

As four edges are incident to each one of the graph's $H$ vertexes $v_h$ (excluding the inner vertexes), then six edges $((4-1)*2)$ are incident to two vertexes, which appear to be the endings of any edge. It turns out that the vertexes of  graph $G$ have the degrees not more than 6. Let us suggest that the Four Color Problem is true and the vertexes of the conjugated triangulation can be colored in three colors. Then each graph's $H$ edge has on its ends the vertexes of different colors. A pair of these colors makes up one of the possible combinations from three colors by two ones. We have 6 such combinations.

So, the following statement becomes evident, which also presents the next hypothesis also equivalent to Four Color Problem.

\subsection {Hypothesis 11}

The edges of the planar conjugated triangulation (graph $H$) can always be painted in six colors in such a way, that no adjacent edges will be painted equally.

But the last hypothesis is true under the Shannon's theorem, which gives us the upper bound for the edge coloring.

So under the Brooks' theorem hypothesis 11 is true, but it appears to be a more strong statement, than the Brooks' theorem, that's why it can be used as the necessary requirement, notably: for the Four Color Problem to be true it is necessary that the edges of the planar conjugated triangulation could be painted in 6 colors.
 
But the question remains: is it possible in all cases and at what conditions is it possible?

Then the sufficient condition can be the following: the edges of the planar triangulation always can be painted in such 6 colors, which represent 6 different combinations from three by two. In its turn the vertexes of the planar triangulation right along can be painted in the specified three colors so, that any pair of the vertexes colors makes the necessary edge color.

In order to prove the last statement we are to examine quite a difficult problem on the existence of such a set of the standard operators, which will always permit to convert three colors of the vertexes of the arbitrary planar conjugated triangulation into 6 colors of its edges and backwards. At that all the requirements on the right coloring must be met.

That's why let us formulate one more hypothesis, which is also equivalent to the Four Color Problem.

\subsection {Hypothesis 12}

The graph's $G(Q,\Gamma)$  vertexes ${q_i }$, which is the adjacency graph of the planar conjugated triangulation $H(Q,V)$, always can be painted in six colors so, that each of the stated 6 colors could be presented as one combination $(C_3^2)$ from three by two, which are generated by the edges of the direct graph $H(Q,V)$ from its vertexes' colors. 

And finally, let us introduce one more hypothesis, which can be formulated this way.

\subsection {Hypothesis 13}

The chromatic number of any planar conjugated triangulation must be equal to not more than 3.

The estimation of the chromatic class of the planar conjugated triangulation can be done with the help of Shannon's theorem \cite {Shannon}.

\begin {theorem} [Shannon's theorem]
If graph $H$ has no loops, then $\chi(H)\leq[3/2 \rho_h(H)]$, where $\chi(H)$ --- the chromatic class of graph $H$ (a number of the colors at the right coloring of the edges); $\rho_h (H)$ --- the maximal degree of graph's $H$ vertexes \cite {Weaver}.
\end {theorem}

Because of the Shannon's theorem for the planar conjugated triangulation, where the degrees of the vertexes are: $\rho_h (H)\leq4$, it follows that its chromatic class is always not more than six: $\chi(H)\leq6$.

Further let us specify the lowest estimation for the chromatic number of the planar conjugated triangulation.

\section {The chromatic number of the planar conjugated triangulation}

We are interested in the coloring of graph's $H$ vertexes. It is known that graph's $H$ edges can be painted with not more than 6 colors. Each edge has on its ends the vertexes of different colors. That is why any of six colors of the edges can be considered as an ordered pair of the vertexes' colors. Let us agree to sort these pairs according to the line of the Euler circuits, which pass along the vertexes and the edges of graph $H$.

So, the color of any inner vertex of graph's $H$ will take part in four ordered pairs of the colors. In two cases this color will stand in the first place (at Euler circuit's exit from the vertex), and in other two cases   in the second place (at Euler circuit's entrance into the vertex). A color of any external vertexes of graph's $H$ will also take part in two ordered pairs: once in the first place, and the other time ---  in the second place. If we identify the colors of the vertexes with the vertexes, we'll be able to say that each vertexes color has an equal number of exits and entrances.

Let us introduce the concept of the abstract $k$-chromatic graph $R$ (Fig. \ref {Fig:2}). Let it be so, that in graph $R_k (V_0,Q_0)$ each vertex $v_\alpha$ corresponds to the subset $\{v_{\alpha i}\}_H$ of graph's $H$ vertexes, which have the color $\alpha$. The degree of each graph's $R_k$ vertex must be the even one and must not be less than the degree of graph's $H$ vertexes. And the degrees of graph's $H$ vertexes are even and they have the following degrees: $\rho(v_H)\leq4$. Each edge (oriented) of graph $R_k$, which connects the vertexes $v_\alpha$ and $v_\beta$ belongs to the subset of the edges, which connect the vertexes of the color $\alpha$ with the vertexes of the color $\beta$ in graph $H$.

\begin{figure}[htb]
		\includegraphics[width=0.4\textwidth]{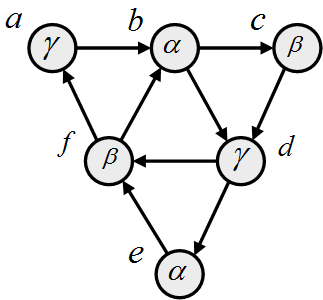}
	\caption{}
	\label{Fig:2}
\end{figure}

Let us construct graph $R_k$, meaning that it appears to be the abstract mapping of graph $H$ and must reproduce its main characteristics \cite {Malinina5}:

\begin {enumerate}

\item
Graph $R_k$ must have a number of the vertexes equal to graph's $H$ chromatic number.
\item
The number of graph's $R_k$ edges must be equal to its chromatic class.
\item
Graph $R_k$ must permit the existence of Euler circuit.
\item
Graph $R_k$ must contain the minimal cycles of length 3.
\item
Matrix $R_k$, also as matrix $P_H$ must have the column $\sum_jP_{ij}$ which is equal to the transposed row $\sum_iP_{ij}$ .

\end {enumerate}

Let us prove the following theorem.

\begin {theorem}
The chromatic number of the abstract graph $R_k$ of the planar conjugated triangulation $H$ always is equal to: $\gamma(R_k)=3$.
\end {theorem}

\begin {proof}
For the proof of the theorem it is sufficient to show that graph $R_k$ possesses the listed above characteristics only in one case --- if the number $k$ of its vertexes is equal to $k=3$. Let us examine the graphs $R_k$ at $k=1,2,3,4$ and two cases of matrix $P_R$:

\begin {itemize}
\item
matrix $P_R$ --- is symmetric one;
\item
matrix $P_R$ --- is antisymmetric one.
\end {itemize}

Under Brooks' theorem the graph's $H$ chromatic number is not more than four, thus it makes no sense to examine graphs $R_k$ at $k>4$.

Let us examine and prove the properties of the $P_R$ matrix, which is the adjacency matrix of graph's $R$ vertexes. Under Shannon's theorem: $\chi(R_k)\leq6$. 

So, accepting all the non-nil elements of matrix $P_R$ as units, we will have: $\sum_i\sum_j{P_{ij}}\leq6$. Matrix $P_R$ must not possess empty rows and columns, so either the row $(\sum_iP_{ij})_j$ or the column $(\sum_jP_{ij})_i$ must not possess nulls.

Hence the coloring of graph's $H$ external vertexes must not provoke any difficulties at the requirement of the parallel coloring of the inner vertexes; therefore these vertexes may enter any graph's $R$ vertexes, which correspond to the submatrixes of the one-color inner vertexes of graph $H$. 

Thus, none of the graph's vertexes can have the degree less than 4. So, both the row $(\sum_iP_{ij})_j$ and the column $(\sum_jP_{ij})_i$ must possess only twains. 

So, $(\sum_iP_{ij})_j=2; 2; 2;…$ and $(\sum_jP_{ij})_i=2;2;2;…$ such as for graph $R_k$ thus for graph $H$. 

Under Shannon's theorem for the graph $R_k$ the condition: 
$$\sum_i\sum_jP_{ij}=\sum_j\sum_iP_{ij}= 6$$
is always true. 

So, $(\sum_iP_{ij})=2;2;2$ and $(\sum_jP_{ij})=2;2;2$. Only the single matrix $P_R$ meets such a requirement (Fig. \ref {Fig:3}).

\begin{figure}[htb]
		\includegraphics[width=0.9\textwidth]{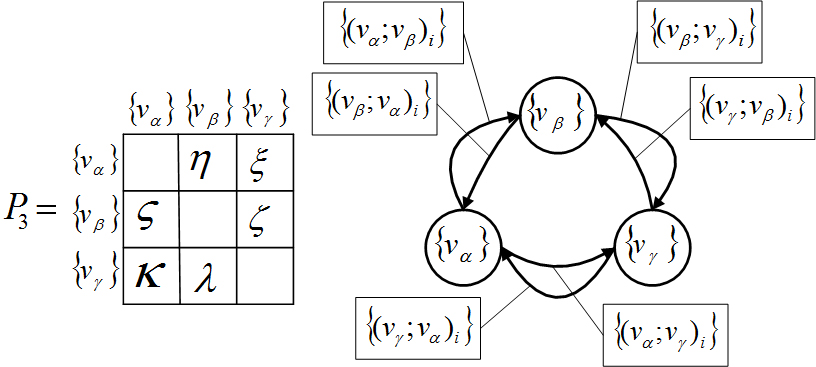}
	\caption{}
	\label{Fig:3}
\end{figure}

This implies that graph's $R$ chromatic number is always equal to 3. Let us confirm the obtained result and \textit{pro hac vice} let us examine the abstract graph $R_k$ at $k=1;2;3;4$.

\textbf{Case 1}. $k=1$. In this case we have null-graph $R_0$, which cannot serve as the abstract mapping of graph $H$.

\textbf{Case 2}. $k=2$. Graph $R_2$ and its matrix are represented in Fig.\ref {Fig:4}.
For graph $R_2$ we have $\rho(R_2 )=2$; $\gamma(R_2 )=2$; $\chi(R_2 )=(3/2)*2=3$. Graph $R_2$ permits the constructing of Euler circuit. Matrix $R_2$ --- is symmetric, so $(\sum_iP_{ij}) _j=(\sum_jP_{ij})_i^T=(1;1)$. But graph $R_2$ cannot serve as the abstract mapping of graph $H$, because it does not permit the construction of the cycles of the length 3.

\begin{figure}[htb]
		\includegraphics[width=0.53\textwidth]{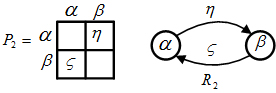}
	\caption{}
	\label{Fig:4}
\end{figure}

\textbf{Case 3}. $k=3$. Graph $R_3$, which is presented in Fig. \ref {Fig:5} with its matrix has 3 vertexes and 6 edges. We have: $\rho(R_3 )=4$; $\gamma(R_3 )=3$; $\chi(R_3)=(3/2)*4=6$.

Graph $R_3$ has a symmetric matrix $P_3$ (Fig. \ref {Fig:5}), so: $(\sum_iP_{ij})_j=(\sum_jP_{ij})_i^T=(2;2;2)$. The graph permits the construction of both Euler circuit and minimal cycles of the length 3. Thus, graph $R_3$ meets the requirements of the theorem.

\begin{figure}[htb]
		\includegraphics[width=0.53\textwidth]{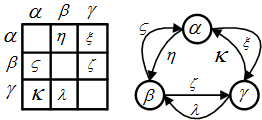}
	\caption{}
	\label{Fig:5}
\end{figure}

\textbf{Case 4}. We have five different graphs $R_k$ with the symmetric matrixes $P_k$ and with the different number of the edges. Let us examine them step by step.

\begin {enumerate}

\item
Graph $R_4^1$ is presented in Fig.\ref{Fig:6}. For this graph we have: $\rho(R_4^1)=6$, but $\rho(H)=4$; $\gamma(R_4^1)=4$ according to the construction; $\chi(R_4^1)=\chi(H)=3/2*4=6$. Its matrix $P_4$ is symmetric, thus $(\sum_iP_{ij})_j=(\sum_jP_{ij})_i^T=(1;3;1;1)$ and $\sum_i\sum_jP_{ij}=6$. 

\begin{figure}[htb]
		\includegraphics[width=0.65\textwidth]{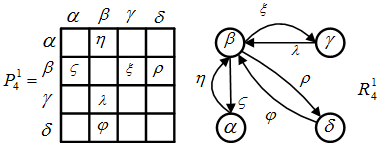}
	\caption{}
	\label{Fig:6}
\end{figure}

The graph permits the creation of Euler circuit. But it does not permit the creation of the minimal cycles of the length 3, so it cannot serve as the mapping of graph $H$.

\item
Graph $R_4^2$ is presented in Fig.\ref{Fig:7}. We have: $\rho(R_4^2)=4$; $\gamma(R_4^2)=4$; $chi(R_4^2)=6$; $(\sum_iP_{ij})_j=(\sum_jP_{ij})_i^T=(1;2;2;1)$; $\sum_i\sum_jP_{ij}=6$. The graph permits the creation of Euler circuit. But it does not permit the creation of the minimal cycles of the length 3. So, graph $R_4^2$ also is excluded and will not be examined.

Let us examine the graphs with the number of the edges more than 6.

\begin{figure}[htb]
		\includegraphics[width=0.65\textwidth]{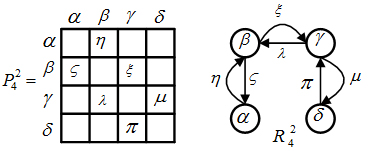}
	\caption{}
	\label{Fig:7}
\end{figure}

\item
Graph $R_4^3$ is presented in Fig. \ref {Fig:8}. We have: $\rho(R_4^3)=6$, but $\rho(H)=4$; $\gamma(R_4^3)=4$; $\chi(R_4^3)=8$ --- according to the construction; $(\sum_iP_{ij})_j=(\sum_jP_{ij})_i^T=(3;2;3;2)$ and $\sum_i\sum_jP_{ij}=10$. The graph permits both the creation of Euler circuit and the creation of the minimal cycles of the length 3. 

\begin{figure}[htb]
		\includegraphics[width=0.65\textwidth]{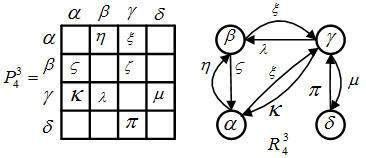}
	\caption{}
	\label{Fig:8}
\end{figure}

But this graph also cannot serve as the mapping of graph $H$ according to the following reasons:

\begin {itemize}

\item	
The vertex $v_\delta$ has $\rho(v_\delta)=2$. Such vertexes in graph $H$ appear to be its external vertexes. But each external vertex is connected with the vertexes of different color by the edges. Thus $v_\delta$ does not belong to a subset of the external vertexes, and so it cannot be the mapping of graph's $H$ vertexes. At the exclusion of the vertex $v_\delta$ together with the edges $\pi$ and $\mu$ graph $R_4^3$ turns into graph $R_3$.

\item
Hence for graph $H$ we have $\rho(H)=4$, then $\chi(H)=6$. Therefore, either graph $R_4^3$ has the wrong edge coloring or at least one vertex $v_\gamma$ has the degree 6. According to both cases graph $R_4^3$ cannot serve as the mapping of graph $H$.

\end {itemize}

\item
Graph $R_4^4$ is presented in Fig. \ref{Fig:9}. We have: $\rho(R_4^4)=6$, but $\gamma(H)=4$; $\gamma(R_4^4)=4$; $\chi(R_4^4)=10$ (according to the construction); $(\sum_iP_{ij})_j=(\sum_jP_{ij})_i^T=(3;2;3;2)$; $\sum_i\sum_jP_{ij}=10$. The graph permits the generation of both Euler circuits and the minimal cycles of the length 3. But the vertex $v_\delta$ may have the equal color with the vertex $v_\beta$, because they are not the adjacent ones. So, due to the accepted method of the edge's coloring: $\mu\equiv\lambda$; $\pi\equiv\epsilon$; $\psi\equiv\nu$; $\omega\equiv\xi$. But in this case graph $R_4^4$ is also transformed into graph $R_3$.

\begin{figure}[htb]
		\includegraphics[width=0.75\textwidth]{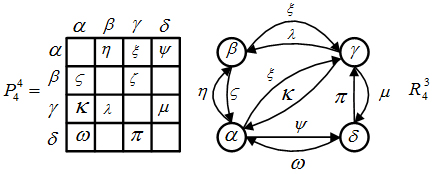}
	\caption{}
	\label{Fig:9}
\end{figure}

\item
Graph $R_4^5$ is presented in Fig. \ref {Fig:10}. We have: $\rho(R_4^5)=6$, but $\rho(H)=4$; $\gamma(R_4^5)=4$; $\chi(R_4^5)=12$ according to the construction; $(\sum_iP_{ij})_j=(\sum_jP_{ij})_i^T=(3;3;3;3)$. 

\begin{figure}[htb]
		\includegraphics[width=0.75\textwidth]{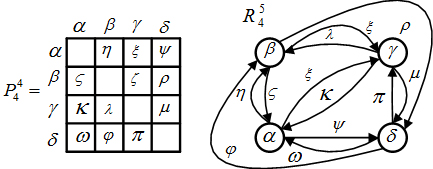}
	\caption{}
	\label{Fig:10}
\end{figure}

The graph permits the creation of both Euler circuits and the minimal cycles of the length 3. But it also cannot serve as graph's $H$ mapping according to the following reasons:

\begin {itemize}
\item
Even if we are able to accept $\rho(P_4^5)=6$, nevertheless it might be: $\chi(P_4^5)=3/2 \rho(P_4^5)=9<12$, and thus the painting of graph $P_4^5$ is wrong. It is a fortiori more wrong in the case of $\rho(H)=4$.
\item
Graph $P_4^5$ permits the presence in graph H vertexes with the degree 6, which is not possible.
\end {itemize}

\end {enumerate}

\end {proof}

Let us examine the possibility of painting the vertexes of graph $R$ in four colors at the condition of painting the edges in 6 colors.

\section {The opportunities for graph's $H$ vertex coloring}

Let us examine the problem of the possibilities of the vertex's coloring in four colors of graph $R$ at the condition of its edges' coloring in six colors. The ordering of vertexes colors will be done along Euler circuit.

Matrix $P_4$ has 6 cells above its diagonal. There exists 20 ways of the distribution of three units in six cells: $C_6^3=6!/3!(6-3)!=20$. And of course there are 20 variants of the symmetric matrixes $P_4$. Five graph $R_4$ modes correspond to these 20 variants; each graph $R_4$ can be oriented in 4 directions (Rig. \ref{Fig:11} and \ref{Fig:12}). Only four graphs (7, 10, 12 and 20) out of 20 can serve as a mapping of graph $H$. All of them are isomorphic with graph $R_3$.

\begin{figure}[htb]
		\includegraphics[width=0.75\textwidth]{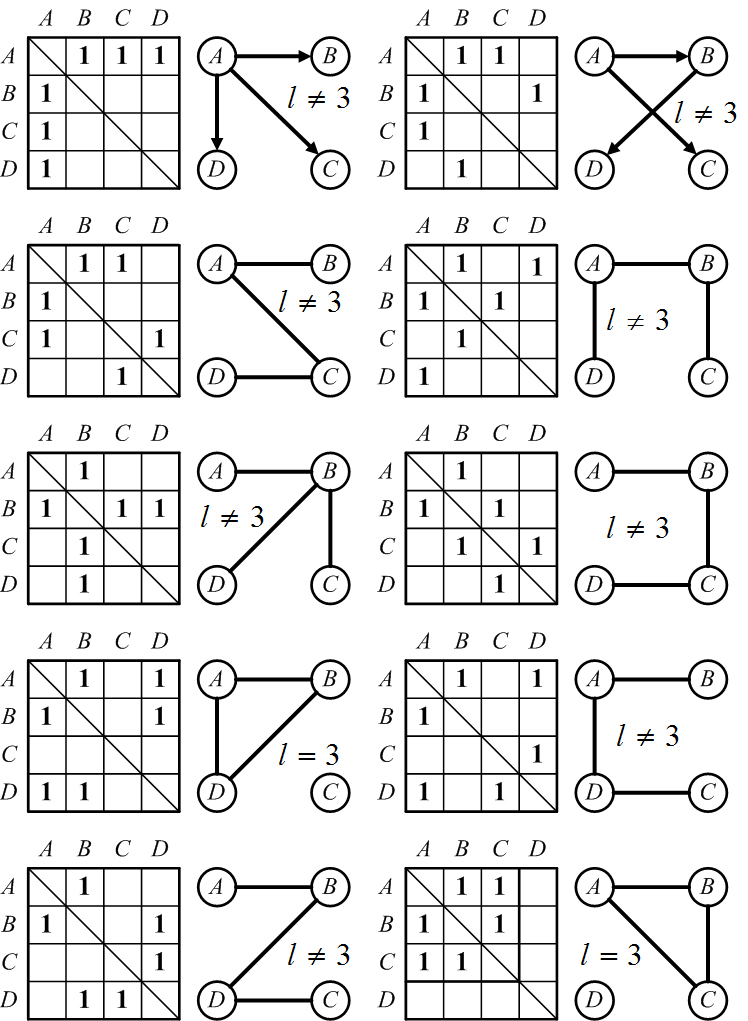}
	\caption{}
	\label{Fig:11}
\end{figure}

\begin{figure}[htb]
		\includegraphics[width=0.75\textwidth]{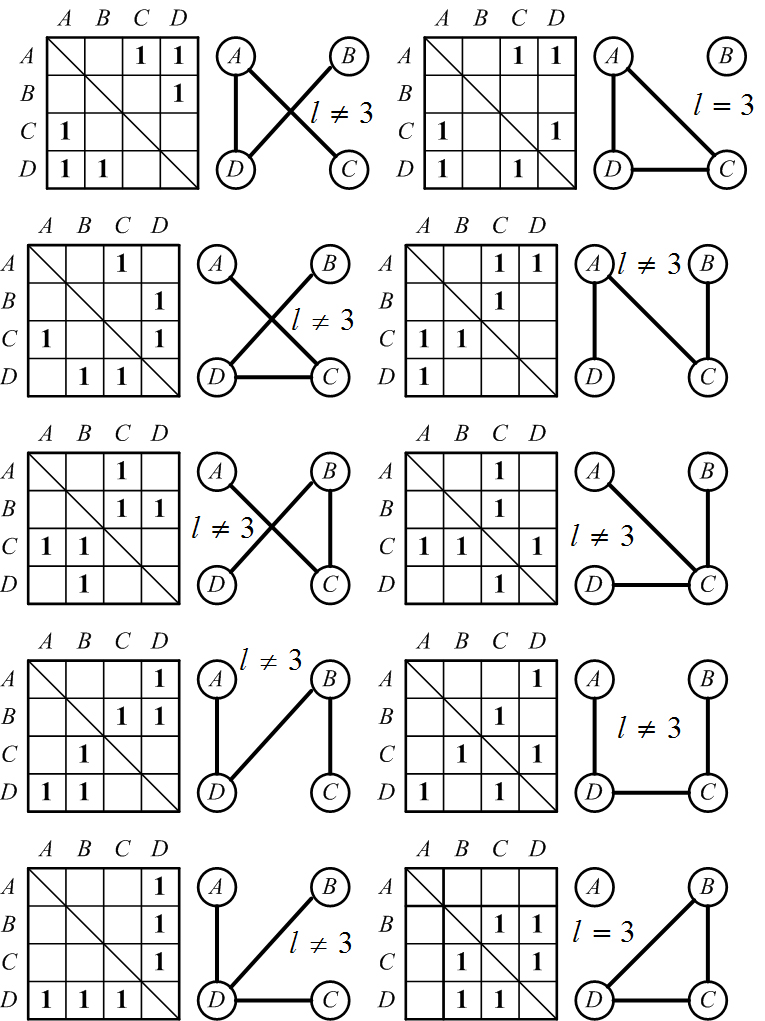}
	\caption{}
	\label{Fig:12}
\end{figure}

But a logical question arises: may some graph $R_4$ exist with the anti-symmetric adjacency matrix, which has $\sum_i\sum_jP_{ij}=6$, $\gamma(R_4)=4$ and which can serve as graph's $H$ mapping?

For the final answer we ought to examine 190 variants of the asymmetrical matrixes $P_4$ and the same amount of graphs $R_4$. But everything might be simpler. Hence it must be Euler circuit in graph $H$ and a number of the vertexes have two exits and two entries, then the row $(\sum_iP_{ij})_j$ and the column $(\sum_jP_{ij})_i$ can contain only numbers 1 and 2 at the condition (requirement) that their sum is equal to 6.

Then there exist $C_4^2$ variants of the rows $(\sum_iP_{ij})_j$, notably: $C_4^2=(2\cdot3\cdot4)/(2\cdot2)=6$ and the same amount of the column's variants $(\sum_jP_{ij})_i$.

As a result, there can exist not more than 21 variants of such matrixes $P_4$, for which both the row $(\sum_iP_{ij})_j$ and the column $(\sum_jP_{ij})_i$ are composed out of both the twins and the units.

For the aim that at least one of these variants may serve as the adjacency graph's $R_4$ matrix (at the requirement that graph $R_4$ serves as the abstract mapping of graph $H$), it is necessary and sufficient for such variant to have term wise sums of the column $(\sum_iP_{ij})_j$ and the column $(\sum_jP_{ij})_i^T$ to meet one of the three following requirements:

\begin {itemize}
\item
It must contain only the groups of fours.
\item
It must contain three fours and one twin.
\item
It must contain three fours and one zero.
\end {itemize}

And none of these is possible.

Thus, the abstract graph $R$ of the planar conjugated triangulation has the right edges' coloring in six colors only iff its chromatic number is equal to three. So, because the right painting of the graph's R edges with the help of six colors is always possible, then the graph $R$ is 3-chromatic.

\begin {corollary}
The chromatic number of the special minimal graph $H_{min}$ is always equal to three: $\gamma(H)=3$.
\end {corollary}

It seems that hereof with the sufficient clearness follows the equity of the Four Color Problem, which now maybe will gain the power of a theorem: \textit{Any planar graph is not more than four-chromatic.}

But can we transfer this statement, also with the sufficient clearness to the general case of the planar triangulation? 

Unfortunately, such statement is not obvious.

\section {The lower estimation of the chromatic number}

Let us specify the low estimation of the chromatic number for the planar conjugated triangulation with the help of the Shannon's theorem.

\begin {theorem}
The chromatic number of any planar conjugated triangulation is situated in the interval: $4\geq(H)\geq3$.
\end {theorem}

\begin {proof}
Euler circuit exists in graph $H(V,Q)$ because the degrees of all its vertexes $v_h$ are even. Let us suggest that all the edges of graph' $H(V,Q)$ are oriented in the direction of the bypass of Euler circuit. From this moment we can examine the vertexes $v_h$ and $v_g$, which assign the beginning and the ending of the arc $q_i$ as the components of some directed pair. Then: $(v_h;v_g)\neq(v_g;v_h)$.

Let us suggest that we have three arbitrary symbols: $\alpha$, $\beta$ and $\gamma$, from which any arbitrary directed pairs can be composed. A number of such pairs is equal to a number of the distributions at two from three: so, it is equal to 6. Indeed, we have: $\alpha\beta$, $\beta\gamma$, $\gamma\alpha$, $\alpha\gamma$, $\gamma\beta$ and $\beta\alpha$, or just 6 directed pairs. All these directed pairs are composed out of three symbols.

Let us examine the edges of graph $H(V,Q)$ as the directed pairs, which are composed out of the vertexes of graph $H(V,Q)$ and let us estimate the coloring of these edges with the help of Shannon's theorem \cite {Weaver}. Concerning our case, the case of the planar conjugated triangulation, under the Shannon's theorem it follows that its chromatic class is not more than six. Indeed, as $\rho_h\leq4$, then $3/2 \rho_h\leq6$.

It hence follows that $\gamma(H)\geq3$. The theorem is proved.

So, the chromatic number of the planar conjugated triangulation meets the requirement: $4\geq\gamma(H)\geq3$.
\end {proof}

\begin {proof}
It is evident that $\gamma(H)\geq2$. Again each edge of the planar conjugated triangulation can be presented as the directed pair of the vertexes. If each new pair is examined as a new color, then under the Shannon's theorem the whole set of graph's $H$ edges must be composed of not more than six subsets of the directed vertexes' pairs. The minimal number of the vertexes' subsets, differing by colors, from which we can compose six subsets of the different pairs, must be equal to 3.

Indeed, out of three subsets ${\alpha}$; ${\beta}$; ${\gamma}$ we can compose 6 subsets of the pairs of these elements: ${\alpha\beta}$; ${\alpha\gamma}$; ${\beta\gamma}$; ${\beta\alpha}$; ${\gamma\alpha}$; ${\gamma\beta}$. It follows that: $\gamma(H)\geq3$. The theorem is proved.
\end {proof}

So, the chromatic number of the planar conjugated triangulation meets the requirement: $4\geq\gamma(H)\geq3$.

But as the six colors always can be presented as 6 distributions from 3 by 2, thus the theorem install the more strict condition (requirement) for the planar conjugated triangulation, than Brook's theorem.

Indeed, if the chromatic number of the planar conjugated triangulation is not more then 4 (under Brook's theorem) then, examining each edge as the ordered pair of the vertexes, being colored into the color, which corresponds to one distribution of 2 by 4, we'll come to the conclusion that a number of such distributions may be significantly more, than it is defined by the upper bound $\chi(H)$, which is declared by Shannon's theorem.

\begin {theorem}
Let graph $G(Q,\Gamma)$ be the adjacency graph of the edges of the planar conjugated triangulation $H(V,Q)$. Then the absolute degrees of graph's $G(Q,\Gamma)$ vertexes are always even and are not more than 6.
\end{theorem}

\begin {proof}
Hence the vertexes' $v_h$ degrees of graph $H(V,Q)$ have the value 2 or 4, and then to every edge in graph $H$ may be adjacent on each end 1 or 3 edges (not more). So, every edge of graph $H$ may have on its both ends 2, 4 or 6 adjacent edges, but not more. So, to each vertex $q_i$ in $G$ also may be incident 2, 4 or 6 edges. That is each graph's $G(Q,\Gamma)$ vertex can have the degree 2, 4 or 6, and not more. The theorem is proved .
\end {proof}

\begin {theorem}
The chromatic number of the graph $G(Q,\Gamma)$ is not more than six.
\end {theorem}

\begin {proof}
The equity of the theorem follows from Shannon's theorem and theorem above.
\end {proof}

Let us formulate the following theorem.

\begin {theorem}
Let $H(V,Q)$ --- be a planar conjugated triangulation, and its chromatic class is $\chi(H)\leq6$. Then its chromatic number is $\gamma(H)\leq3$.
\end {theorem}

\begin {proof}
Let us do the proof in the mode of the equality: $\gamma(H)=3$. It will be sufficient for the proof of hypothesis 11, and, therefore for the proof of Four Color Problem.

Let us again construct graph $G(Q,\Gamma)$, which will serve as the edge adjacency graph of the planar conjugated triangulation $H$. The construction will be made in the following way: we'll sign the centers of all graph's $H$ edges with the dots; we'll accept these dots as the vertexes of graph $G$; we'll connect the dots with the help of the edges in such cases, when edges $q_i$, which are appropriate to the vertexes $q_i$ in graph $H$ have the common vertex (Fig. \ref{Fig:1}). A subgraph of the triangulation $H$ is cut out by the dashed lines and is represented in Fig. \ref{Fig:13}. The appropriate subgraph $G$ is presented in Fig. \ref{Fig:15}. 

\begin{figure}[htb]
		\includegraphics[width=0.75\textwidth]{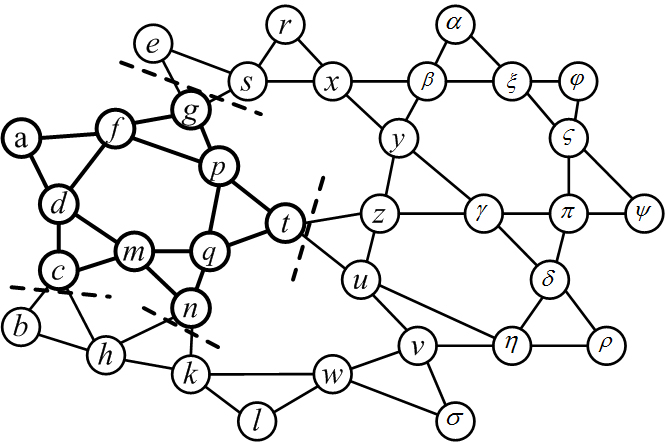}
	\caption{}
	\label{Fig:13}
\end{figure}

\begin{figure}[htb]
		\includegraphics[width=1\textwidth]{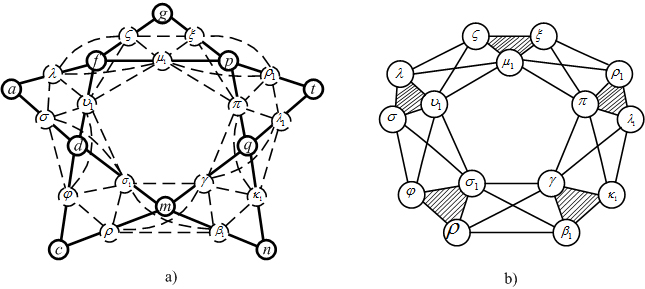}
	\caption{}
	\label{Fig:15}
\end{figure}

The subgraph $H_1$ in Fig.\ref{Fig:15} on the left is presented by thin lines and the subgraph $G_1$ -- by the dashed lines. Fig.\ref{Fig:15} on the right represents only graph $G_1$.

\end {proof}

Let us prove lemma.

\begin {lemma}
Graphs $G$ vertexes' degrees (the adjacency edge graph of the planar conjugated triangulation $H$) are not more than six.
\end {lemma}

\begin {proof}
As the vertexes' $v_h$ degrees of graph $H(V,Q)$ have the values 2 or 4, then to each edge of $H$ there may be adjacent on one end 1 or 3 edges, and on the other end also 1 or 3 edges. So, each graph's $H$ edge can have at the same time 2, 4 or 6 adjacent edges, but not more. 
\end {proof}

The equity of the next theorem follows from Brooks' theorem and the previous lemma.

\begin {theorem}
The chromatic number of the graph $G$ is not more than six: $\chi(G)\leq6$.
\end {theorem}

Let us return to the theorem on the chromatic number of the planar conjugated triangulation. The proof's complexity of this theorem consists of the fact that six colors of graph's $G$ vertexes can be examined as 6 distributions from 3 by 2. The conclusion, that for the coloring of graph's $H$ vertexes three colors will be always enough, seems naturally. But from the statement that it will be enough 3 initial colors for the creating of 6 colors, the conclusion, that 6 ordered pairs can be combined out of three initial ones, does not follow directly. Six ordered pairs can be composed out of 3, 4 and more initial elements so that in each case any component may be included at least into one pair (Fig.\ref{Fig:16}). Everything depends on the mode of the ordering of those elements, from which the ordered pairs are composed.

\begin{figure}[htb]
		\includegraphics[width=1\textwidth]{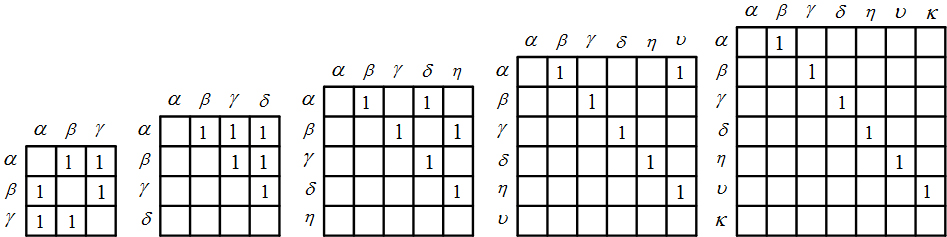}
	\caption{}
	\label{Fig:16}
\end{figure}

For example, in compliance with Brook's theorem, four colors are enough for the right coloring of graph H. Six ordered pairs, composed out of 4 elements, can be easily introduced:

\begin {itemize}

\item
The initial elements: $\alpha$, $\beta$, $\gamma$, $\delta$;
\item
The ordered pairs: $\alpha\beta$, $\alpha\gamma$, $\alpha\delta$, $\beta\gamma$, $\beta\delta$, $\gamma\delta$.

\end {itemize}

Our task is to prove that it is always possible that the vertexes of graph $H$ are to be colored in three colors for the formation of six colors for the vertexes of graph $G$ as the ordered heterochromous (varicolored) ones. At that no adjacent vertexes of graph $H$ will be colored into one color, and no vertexes will remain not colored.

Otherwise, it is necessary to prove a possibility of the existence in graphs $G$ and $H$ of such operators (functional) of the inverse transformation of graph's $H$ vertexes' colors into the colors of graph's $G$ vertexes. It is necessary to prove that for the functional embedding or implementation for the vertexes of graph $G$ to be colored in six colors it is always enough that the vertexes of graph $H$ are to be colored in three colors. It is also necessary to prove that by the proper choice of the operators it is always possible to transform the six-color painting of graph's $G$ vertexes into the three-color painting of graph's $H$ vertexes.

Let us also demand that these operators allowed to transform any pair from six into another five and backwards.

Let us at first examine the problem of such possible methods of the coloring of graph's $G$ vertexes into six colors, starting from the coloring of graph $H$, which will need exactly three colors for the vertexes of graph $H$. Let us show that if graph's $H$ vertexes can be colored in three colors, then there exist exactly 16 variants of coloring of any six vertexes, which are the adjacent ones with the given vertex, which is colored in one of six colors, into 5 colors.

Let us further examine some methods for the coloring of the vertexes of the abstract graph $G_6$ on the base of the painting of graph's $H_6$ vertexes, which can be wittingly painted in three colors. These methods will be presented by the finite number of the algorithmic operators. And finally we'll show that all the methods above aimed at the coloring of graph's G vertexes in six colors can be implemented by the operators of graph's $G_6$ edges, which require the painting of graph's $H_6$ vertexes in three colors.
All this will be represented a little later in the following papers.

\section {Conclusions:}

\begin {enumerate}

\item
The chromatic number estimation for the planar conjugated triangulation allows introducing the new hypothesis, which is equivalent to Four Color Problem: the vertexes of the planar conjugated triangulation can be colored in three colors thus no adjacent vertexes will be co-lored equally.
\item
The introduction of the next dual graph allows bringing in another three new hypotheses, which are also equivalent to Four Color Problem:

\begin {enumerate}
\item
The edges of the planar conjugated triangulation (graph $H$) can always be colored in six colors in such a way, that no adjacent edges will be colored equally.
\item
Graph's $G(Q,\Gamma)$ is the adjacency graph of the planar conjugated triangulation $H(Q,V)$. Its vertexes $\{q_i\}$ always can be colored in six colors so, that each of the stated 6 colors can be presented as one combination $(C_3^2)$ from three on two (by pair), which are generated by the edges of the direct graph $H(Q,V)$ from its vertexes' colors.
\item
The chromatic number of any planar conjugated triangulation must be equal to not more than three.
\end {enumerate}
\item
The introduction of the abstract minimal graph gives us a possibility to prove some theorems:

\begin {enumerate}
\item
The abstract graph $R_k$ of the planar conjugated triangulation $H$ always has the chromatic number equal to: $\gamma(R_k )=3$.
\item
The examination of the properties of the vertex matrix of graph $R$ proves that it can exist at the number of the vertexes equal to 4.
\item
The conditions for the painting graph's $R$ vertexes in 3 colors at the requirement of its edge painting in 6 colors.
\end {enumerate}

\item
It is proved that the chromatic number of the abstract minimal graph $H_{min}$ is always equal to 3: $\gamma(H_{min})=3$.
\item
It is proved that the chromatic number of the planar conjugated triangulation is situated in the interval: $3\leq\gamma(H)\leq4$.
\item
The properties of the vertexes of graph $G$ are proved: they are the even ones and their degrees are not more than 6.
\item
It is proved that the chromatic number of graph $G$ is equal to 6.

\end {enumerate}

\section {Acknokledgments:}

A lot of thanks for those members of my family who undergone all the difficulties side by side with me and who encouraged me in my work.  

Also I will be very grateful to those readers, who will find and send me a word about the uncovered misprints or some errors in order to improve the text.


\end{document}